\documentclass[11pt]{article}

\usepackage[utf8]{inputenc}

\DeclareUnicodeCharacter{1EC5}{}

\usepackage{algorithm,algpseudocode}
\usepackage{amsmath}
\usepackage{amsthm}
\usepackage{amssymb}
\usepackage{graphicx}
\usepackage{comment}
\usepackage{subfig}
\usepackage[table]{xcolor}
\usepackage{booktabs}
\usepackage{color}   
\usepackage{hyperref}
\hypersetup{
    colorlinks=true, 
    linktoc=all,     
    linkcolor=blue,
}

\usepackage{todonotes}

\usepackage[
backend=biber,
style=alphabetic,
sorting=ynt
]{biblatex}
\newtheorem{thm}{Theorem}
\newtheorem{observation}[thm]{Observation}
\newtheorem{corollary}[thm]{Corollary}
\newtheorem{definition}[thm]{Definition}
\newtheorem{lem}[thm]{Lemma}
\newtheorem{prop}[thm]{Proposition}

\newtheorem{problem}[thm]{Problem}

\newcommand{\poly}{\text{poly}}

\newcommand{\norm}[2]{\left\|#1\right\|_{#2}}
\newcommand{\inner}[2]{\left\langle #1, #2 \right\rangle}
\newcommand{\abs}[1]{\left|#1\right|}
\newcommand{\E}{\mathbb{E}}

\newcommand{\R}{\mathbb{R}}
\newcommand{\parens}[1]{\left(#1 \right)}

\newcommand{\eps}{\epsilon}
\newcommand{\tildeO}{\widetilde{O}}

\renewcommand{\E}{\mathbb{E}}

\DeclareMathOperator{\Tr}{Tr}
\DeclareMathOperator{\diag}{diag}

\DeclareMathOperator{\spn}{span}

\DeclareMathOperator{\tv}{TV}

\usepackage[margin=1in]{geometry}

\date{}

\begin{document}
\title{Optimal Eigenvalue Approximation via Sketching}
\date{}
\author{
William Swartworth
\\UCLA \and
David P. Woodruff\\CMU
}
\maketitle
\thispagestyle{empty}


\begin{abstract}
    Given a symmetric matrix $A$, we show from the simple sketch $GAG^T$, where $G$ is a Gaussian matrix with $k = O(1/\epsilon^2)$ rows, that there is a procedure for approximating all eigenvalues of $A$ simultaneously to within $\epsilon \|A\|_F$ additive error with large probability.  Unlike the work of (Andoni, Nguyen, SODA, 2013), we do not require that $A$ is positive semidefinite and therefore we can recover sign information about the spectrum as well. Our result also significantly improves upon the sketching dimension of recent work for this problem (Needell, Swartworth, Woodruff FOCS 2022), and in fact gives optimal sketching dimension. Our proof develops new properties of singular values of $GA$ for a $k \times n$ Gaussian matrix $G$ and an $n \times n$ matrix $A$ which may be of independent interest. Additionally we achieve tight bounds in terms of matrix-vector queries. Our sketch can be computed using $O(1/\epsilon^2)$ matrix-vector multiplies, and by improving on lower bounds for the so-called rank estimation problem, we show that this number is optimal even for adaptive matrix-vector queries.
\end{abstract}

\clearpage

\setcounter{page}{1}

\newpage


\section{Introduction}
Estimating the eigenvalues of a real symmetric matrix has numerous applications in data analysis, engineering, optimization, spectral graph theory, and many other areas. As modern matrices may be very large, traditional algorithms based on the singular value decomposition (SVD), subspace iteration, or Krylov methods, may be be too slow. Therefore, a number of recent works have looked at the problem of creating a small summary, or sketch of the input matrix, so that from the sketch one can approximate each of the eigenvalues well. Indeed, in the realm of sublinear algorithms, this problem has been studied in the streaming model \cite{andoni2013eigenvalues}, the sampling and property testing models \cite{BLW019,bcj20,bddmr21,BKM22}, and matrix-vector and vector-matrix-vector query models \cite{andoni2013eigenvalues,LNW14,LNW19,nsw22}; the latter model also contains so-called bilinear sketches. 

In this work we focus on designing linear sketches for eigenvalue estimation. Namely, we are interested in estimating the spectrum of a real symmetric matrix $A \in \R^{n\times n}$ up to $\eps\norm{A}{F}$ error via a bilinear sketch $G A G^T$ with $G\in \R^{k\times n}$ is a matrix of i.i.d. $N(0,1/k)$ random variables, i.e., Gaussian of mean zero and variance $1/k$. The algorithm should succeed with large constant probability in estimating the entire spectrum. This is a very natural sketch, and unsurprisingly has been used before both in \cite{andoni2013eigenvalues} to estimate eigenvalues with an additive error of roughly $\epsilon \sum_{i=1}^n |\lambda_i(A)|$, where $\lambda_i(A)$ are the eigenvalues of $A$, as well as in \cite{nsw22} for testing if a matrix is positive semidefinite (PSD). We note that the additive error of $\epsilon \|A\|_1 = \epsilon \sum_{i=1}^n |\lambda_i(A)|$ can be significantly weaker than our desired $\eps \norm{A}{F}$ error, as $\norm{A}{F}$ can be as small as $\frac{\|A\|_1}{\sqrt{d}}$. This is analogous to the $\ell_2$ versus $\ell_1$ guarantee for heavy hitters in the data stream model, see, e.g., \cite{W16}. 

It may come as a surprise that $GAG^T$ has any use at all for achieving additive error in terms of $\epsilon \|A\|_F$! Indeed, the natural way to estimate the $i$-th eigenvalue of $A$ is to output the $i$-th eigenvalue of $GAG^T$, and this is exactly what the algorithm of \cite{andoni2013eigenvalues} does. However, by standard results for trace estimators, see, e.g., \cite{MMMW21} and the references therein, the trace of $GAG^T$ is about the trace of $A$, which can be a $\sqrt{d}$ factor larger than $\|A\|_F$, and thus the estimation error can be much larger than $\epsilon \|A\|_F$. This is precisely why \cite{andoni2013eigenvalues} only achieves additive $\epsilon \|A\|_1$ error with this sketch. Moreover, the work of \cite{nsw22} does use sketching for eigenvalue estimation, but uses a different, and much more involved sketch based on ideas for low rank approximation of PSD matrices \cite{CW17}, and achieves a much worse $\tilde{O}(k^2/\epsilon^{12})$ number of measurements to estimate each of the top $k$ eigenvalues, including their signs, up to additive error $\epsilon \|A\|_F$. Here we use $\tilde{O}()$ notation to suppress poly$(\log(n/\epsilon))$ factors. Note that for $k > 1/\epsilon^2$, one can output $0$ as the estimate to $\lambda_k$, and thus the sketch size of \cite{nsw22} is $\tilde{O}(1/\epsilon^{16})$. 

To achieve error in terms of $\|A\|_F$, the work of \cite{andoni2013eigenvalues} instead considers the sketch $GAH^T$, where $G,H \in \R^{k \times n}$ are independent Gaussian matrices. However, the major issue with this sketch is it inherently loses sign information of the eigenvalues. Indeed, their algorithm for reconstructing the eigenvalues uses only the sketched matrix, while forgetting $G$ and $H$ (more specifically they only use the singular values of this matrix).  However the distributions of $G$ and $H$ are invariant under negation, so the sketch alone cannot even distinguish $A$ from $-A.$ In addition to this, even if one assumes the input $A$ is PSD, so that the signs are all positive, 
%
their result for additive error $\epsilon \|A\|_F$ would give a suboptimal sketching dimension of $k = \tilde{O}(1/\eps^3)$; see further discussion below. 

\subsection{Our Contributions}

\paragraph{Optimal Sketching Upper Bound.} We obtain the first optimal bounds for eigenvalue estimation with the natural $\epsilon \|A\|_F$ error via sketching. We summarize our results compared to prior work in Table \ref{tab:results}. We improve over  \cite{andoni2013eigenvalues,nsw22} in the following crucial ways.

\begin{table}[ht]
\caption{Our work and prior work on estimating each eigenvalue of an arbitrary symmetric matrix $A$ up to additive $\epsilon \|A\|_F$ error.}
\label{tab:results}
\centering
\begin{tabular}{ c c c c c c c } 
\toprule 
& Sketching dimension & Reference & Notes \\
\midrule
& $\tilde{O}(1/\epsilon^6)$ & \cite{andoni2013eigenvalues} & Loses sign information \\
& $\tilde{O}(1/\epsilon^{16})$  & \cite{nsw22} & \\
& $\Omega(1/\epsilon^4)$ & \cite{nsw22} & Lower bound\\
& $O(1/\epsilon^4)$ & {\bf Our Work} &  \\
\bottomrule
\end{tabular}
\end{table}

Qualitatively, we drop the requirement that $A$ is PSD. As mentioned, the eigenvalues of our sketch $GAG^T$ may not be good approximations to the eigenvalues of $A$. In particular, we observe that the sketched eigenvalues concentrate around $\frac1k \Tr(A),$ which could be quite large, on the order of $\frac{\sqrt{d}}{k}\norm{A}{F}$.  By shifting the sketched eigenvalues by $-\frac1k \Tr(A)$ via an additional trace estimator we compute, this enables us to correct for this bias, and we are able to show that the resulting eigenvalues are good approximations to those of $A.$  In order to perform this correction we in fact require the sketched eigenvalues to concentrate around $\frac1k \Tr(A)$.  Obtaining this concentration is where we require Gaussianity in our argument\footnote{However in the appendix we give a faster sketch for PSD matrices.}.  We leave it as an open question to obtain similar concentration from common sketching primitives.

\paragraph{Comparison with existing work.} Quantitatively, the analysis of \cite{andoni2013eigenvalues} for the related $GAH^T$ sketch works by splitting the spectrum into a ``head" containing the large eigenvalues, and a ``tail" containing the remaining eigenvalues.  The authors then incur an additive loss from the operator norm of the tail portion of the sketch, and show that the head portion of the sketch approximates the corresponding eigenvalues to within a multiplicative error.  Notably, their multiplicative constant is uniform over the large eigenvalues.  This is a stronger guarantee than we need.  For example, to approximate an eigenvalue of $1/2$ to within $\eps$ additive error, we need a $(1\pm O(\eps) )$ multiplicative guarantee.  However to approximate an eigenvalue of $2\eps$ to within $\eps$ additive error, a $(1\pm O(1))$ multiplicative guarantee suffices.  In other words, smaller eigenvalues require less stringent multiplicative guarantees to achieve the same additive guarantee. We leverage this observation in order to get a uniform \textit{additive} guarantee for the large eigenvalues, while not relying on a uniform multiplicative guarantee. Thus, we improve the worst-case $k = O(1/\epsilon^3)$ bound of \cite{andoni2013eigenvalues} to a $k = O(1/\epsilon^2)$ bound for an $\epsilon \|A\|_F$ error guarantee. 

Indeed, one can show if the eigenvalues of $A$ are, in non-increasing order, \[\frac{c_d}{\sqrt{1}}, \frac{c_d}{\sqrt{2}}, \frac{c_d}{\sqrt{3}}, \frac{c_d}{\sqrt{4}}, \ldots, \frac{c_d}{\sqrt{d}},\] 
where $c_d = O(\log^{-1/2} d)$ so that $\norm{A}{F}=1$, then $O(1/\epsilon^3)$ is the bound their Theorem 1.2 and corresponding Lemma 3.5 would give. To see this, their Lemma 3.5, which is a strengthening of their Theorem 1.2, states that for $i=1\ldots k,$
\begin{equation}
    \label{eq:andoni_lemma}
    \abs{\lambda_i^2(GAH^T) - \lambda_i^2(A)} \leq \alpha\lambda_i^2(A) + O\left(\lambda_k^2(A)\right) + O\left(\frac{\alpha^2}{k} \norm{A_{-k}}{F}^2\right),
\end{equation}
with sketching dimension $O(k/\alpha^2)$ on each side (and hence $O(k^2/\alpha^4)$ total measurements). Suppose $\norm{A}{F}=O(1)$ and that we would like to use this bound to approximate $\lambda_{\ell}(A) > \alpha$ to within $\eps$ additive error. After adjusting for the squares, this is equivalent to bounding the left-hand side of (\ref{eq:andoni_lemma}) by $O(\eps \lambda_{\ell})$ for $i=\ell.$  Obtaining such a bound from (\ref{eq:andoni_lemma}) requires that the first two terms on the right-hand side are bounded by $O(\eps \lambda_{\ell}(A))$, i.e., that $\alpha \leq O(\eps/\lambda_{\ell}(A))$ and $\lambda_k^2(A) \leq O(\eps \lambda_{\ell}(A))$.  For the spectrum above, we must therefore take $k\gtrsim c_d \frac{\sqrt{\ell}}{\eps},$ which results in a sketching dimension of 
\[
\frac{k}{\alpha^2} \approx \frac{c_d \sqrt{\ell}}{\eps}\cdot \frac{\lambda_{\ell}(A)^2}{\eps^2} = \frac{c_d^3}{\eps^3 \sqrt{\ell}}
\]
on each side.

Thus for this spectrum, \cite{andoni2013eigenvalues} requires a sketching dimension of $O(1/\eps^3)$ (up to $\log d$ factors) to approximate the largest eigenvalues of $A$ to $\eps$ additive error. Indeed this bound does not achieve $O(1/\eps^2)$ sketching dimension, unless $\ell \gtrsim 1/\eps^2$, at which point $\lambda_{\ell}(A) \leq O(\eps)$ and does not need to be approximated by our algorithm.


We note that while \cite{nsw22} could also report the signs of the approximate eigenvalues, their $\tilde{O}(1/\epsilon^{16})$ sketch size makes it considerably worse for small values of $\epsilon$. 

In contrast, our sketching dimension $k$ is optimal among all non-adaptive bilinear sketches, due to the proof of part 1 of Theorem 31 of \cite{nsw22} applied with $p = 2$. Indeed, the proof of that theorem gives a pair of distributions on matrices $A$ with $\|A\|_F = \Theta(1)$ for which in one distribution $A$ is PSD, while in the other it has a negative eigenvalue of value $-\Theta(\epsilon)$. That theorem shows $\Omega(1/\epsilon^4)$ non-adaptive vector-matrix-vector queries are required to distinguish the two distributions, which implies in our setting that necessarily $k = \Omega(1/\epsilon^2)$. 

\paragraph{Concentration of Singular Values with Arbitrary Covariance Matrices.} Of independent technical interest, we give the first bounds on the singular values of $GB$ for an $n \times n$ matrix $B$ and a (normalized) Gaussian matrix $G$ with $k$ rows when $k \ll n$. When taken together, our upper and lower bounds on singular values show for any $1 \leq \ell$ and  $k \geq \Omega(\ell)$, that 
\begin{eqnarray}\label{eqn:singular}
\sigma_{\ell}(GB)^2 = \sigma_{\ell}(B)^2 \pm O \left (\frac{1}{\sqrt{k}} \right) \norm{B}{F}^2.
\end{eqnarray}
Although there is a large body of work on the singular values of $GB$, to the best of our knowledge there are no quantitative bounds of the form above known. There is work upper bounding $\|GB\|_2$ for a fixed matrix $B$ \cite{vershynin2011spectral}, and classical work (see, e.g., \cite{vershynin2010introduction}) which bounds all the singular values of $G$ when $B$ is the identity, but we are not aware of concrete bounds that prove concentration around $\|GB\|_F^2$ of the form in (\ref{eqn:singular}) for general matrices $B$ that we need. 

\paragraph{Optimal Adaptive Matrix-Vector Query Lower Bound.} A natural question is whether adaptivity can further reduce our sketching dimension. We show that at least in the matrix-vector product model, where one receives a sequence of matrix-vector products $Av^1, Av^2, \ldots, Av^r$ for query vectors $v^1, v^2, \ldots, v^r$ that may be chosen adaptively as a function of previous matrix-vector products, that necessarily $r = \Omega(1/\epsilon^2)$.

Note that our non-adaptive sketch $GAG^T$ gives an algorithm in the matrix-vector product model by computing $AG^T$, and so $r = k = O(1/\epsilon^2$). This shows that adaptivity does not help for eigenvalue estimation, at least in the matrix-vector product model. 

Our hard instance is distinguishing a Wishart matrix of rank $r$ from a Wishart matrix of rank $r+2$ (the choice of $r+2$ rather than $r+1$ is simply for convenience). We first argue that for our pair of distributions, adaptivity does not help. This uses rotational invariance properties of our Wishart distribution, even conditioned on the query responses we have seen so far. In fact, our argument shows that without loss of generality, the optimal tester is a non-adaptive tester which just observes the leading principle submatrix of the input matrix $A$. We then explicitly bound the variation distance between the distributions of a Wishart matrix of rank $r$ and one of rank $r+2$. We also give an alternative, but related proof based on distinguishing a random $r$ dimensional subspace from a random $r+2$ dimensional subspace, which may be of independent interest. As an example, we note that this lower bound immediately recovers the $\Omega(1/\eps)$ matrix-vector lower bound for estimating the trace of a PSD matrix to within $(1\pm \eps)$ multiplicative error \cite{meyer2021hutch++,JPWZ21}, as well as the $\Omega(1/\eps^p)$ lower bound given in \cite{woodruff2022optimal} for approximating the trace of $A$ to additive $\eps \norm{A}{p}$ error (however the bound in \cite{woodruff2022optimal} is more refined as it captures the dependence on failure probability).

These results substantially broaden a previous lower bound for the rank-estimation problem \cite{sun2021querying}. Whereas the hard instance in \cite{sun2021querying} requires  some non-zero eigenvalues to be extremely small, we show that the rank estimation problem remains hard even when all nonzero eigenvalues have comparable size (or in fact, even when they are all equal).


\subsection{Additional Work on Sampling in the Bounded Entry Model}
%
%
%
Recent work has considered the spectral estimation problem for entry queries to bounded-entry matrices. The work of \cite{bddmr21} gives an $\tildeO(1/\eps^6)$ query algorithm for approximating all eigenvalues of a symmetric matrix to within $\eps \norm{A}{F}$ additive error, given a row-norm sampling oracle.  However it remains open whether this bound can be improved to $\tildeO(1/\eps^4)$ even for principal submatrix queries.

Our result shows that $O(1/\eps^4)$ queries is at least attainable under the much less restrictive model of vector-matrix-vector queries.  In contrast to \cite{bddmr21}, our algorithm does not simply return the eigenvalues of our sketch. Indeed no such algorithm can exist as it would violate the one-sided lower bound of \cite{nsw22}.

%
%

\section{Sketching Algorithm and Proof Outline}

\begin{algorithm}[hbt!]
\caption{}\label{alg:main}
\begin{algorithmic}
\Require $A\in \R^{d\times d}$ real symmetric, $k\in \mathbb{N}.$
\Procedure{spectrum\textunderscore appx}{$A$,$k$}
\State Sample $G\in \R^{k\times k}$ with i.i.d. $\mathcal{N}(0,1/k)$ entries.
\State $S \leftarrow G A G^T$
\State For $i=1, \ldots, k$, let $\alpha_i = \lambda_i(S) - \frac{1}{k}\Tr(S)$
\State For $i=k+1,\ldots, d$, let $\alpha_i=0$
\State \Return $\alpha_1, \ldots, \alpha_d$ sorted in decreasing order
\EndProcedure
\end{algorithmic}
\end{algorithm}

\begin{thm}
\label{thm:spectral_apx}
Let $A\in \R^{d\times d}$ be symmetric (not necessarily PSD) with eigenvalues $\lambda_1 \geq \ldots \geq \lambda_d$. For $k\geq \Omega(1/\eps^2)$, Algorithm~\ref{alg:main} produces a sequence $(\mu_1,\ldots, \mu_d)$ such that $\abs{\mu_i - \lambda_i} < \eps \norm{A}{F}$ for all $i$ with probability at least $3/5.$
\end{thm}

\subsection{Proof Outline}
A natural idea is to split the spectrum of $A$ into two pieces, $A_1$ and $A_2$, where $A_1$ consists of the large eigenvalues of $A$ which are at least $\eps \norm{A}{F}$ in magnitude, and where $A_2$ contains the remaining spectral tail.  The eigenvalues of $G A_2 G^T$ will all concentrate around $\Tr(A)$ up to $O(\eps)$ additive error.

We are then left with showing that the eigenvalues of $GA_1G^T$ are $O(\eps)$ additive approximations to the nonzero eigenvalues of $A_1.$  In order to do this we prove upper and lower bounds on the eigenvalues of $GA_1G^T$.  For the upper bound (or lower bound if $\lambda_{\ell}(A_1)$ is negative) we give a general upper bound on the operator norm of $GMG^T$ for a PSD matrix $M$ with $\norm{M}{F}\leq 1.$ By applying this result to various deflations of $A_1$ we are able to give an upper bound on all eigenvalues of $A_1$ simultaneously.

For the lower bound, we first prove the analogous result in the PSD case where it is much simpler. We then upgrade to the general result.  To get a lower bound on $\lambda_{\ell}(GDG^T)$ in the general case, we construct an $\ell$ dimensional subspace $S_{\ell}$ so that $u^T GDG^T u$ is large for all unit vectors $u$ in $S_{\ell}.$  A natural choice would be to take  $S_{\ell}$ to be the image of $GD_{+,\ell}G^T,$ where $D_{+,\ell}$ refers to $D$ with all but the top $\ell$ positive eigenvalues zeroed out.  We would then like to argue that the quadratic form associated to $GD_{-}G^T$ is small in magnitude uniformly over $S_{\ell}$.  Unfortunately it need not be as small as we require, due to the possible presence of large negative eigenvalues in $D_{-}.$  We therefore restrict our choice of $S_{\ell}$ to lie in the orthogonal complement of the largest $r$ negative eigenvectors of $GD_{-}G^T$.  Since we restrict the choice of $S_{\ell}$ we incur a cost, which damages our lower bound on $\lambda_{\ell}(GD_{+}G^T)$ slightly.  However by choosing $r$ carefully, we achieve a lower bound on $\lambda_{\ell}(GDG^T)$ of $\lambda_{\ell}(D) - O(\eps).$

\section{Proof of Theorem~\ref{thm:spectral_apx}}

In this section and the next, we provide upper and lower bounds on the eigenvalues of a sketched $d\times d$ matrix.  We emphasize the results below will later be applied only to the matrix $A_1$ which is rank $O(1/\eps^2).$  Hence we will use the results below for $d=O(1/\eps^2).$

\subsection{Upper bounds on the sketched eigenvalues}

The following result is a consequence of Theorem 1 in \cite{cohen2015optimal} along with the remark following it.
\begin{thm}
\label{thm:nelson_woodruff}
Let $G\in \R^{m\times n}$ have i.i.d. $\mathcal{N}(0,1/m)$ entries, and let $A$ and $B$ be arbitrary matrices with compatible dimensions.  With probability at least $1-\delta$,
\[\norm{A^T G^T G B - A^T B}{} \leq \eps \sqrt{\norm{A}{}^2 + \frac{\norm{A}{F}^2}{k}} \sqrt{\norm{B}{}^2 + \frac{\norm{B}{F}^2}{k}},\]

for $m = O(\frac{1}{\eps^2}(k + \log\frac{1}{\delta}))$. 
\end{thm}

\begin{lem}
\label{lem:op_norm_upper_bound}
Let $D\in \R^{d\times d}$ have eigenvalues $\lambda_1 \geq \ldots \geq \lambda_d \geq 0$ where $\norm{D}{F}\leq 1$.  Let $G\in \R^{t\times d}$ have $\mathcal{N}(0,1/t)$ entries.  The bound
\[\norm{GD^{1/2}}{}^2 \leq \lambda_1 + O\left(\frac{1}{\sqrt{m}}\right)\]
holds with probability at least $1-\frac{1}{20}2^{-\min(m,1/\lambda_1^2)},$ provided that $t\geq \Omega(m + d).$
\end{lem}

\begin{proof}
We first decompose $D$ into two parts $D = D_1 + D_2$ where $D_1$  contains the eigenvalues of $D$ larger than $\lambda_1/2$ and $D_2$ contains the eigenvalues which are at most $\lambda_1/2.$  Let $x$ be an arbitrary unit vector and partition its support according to $D_1$ and $D_2$ so that $x = x_1 + x_2$.  This allows us to write
\begin{align*}
    x^T D^{1/2} G^T G D^{1/2} x 
    &= x_1^T D_1^{1/2} G^T G D_1^{1/2} x_1 + x_2^T D_2^{1/2} G^T G D_2^{1/2} x_2 \\
    &\hspace{0.5cm} + 2 x_1^T D_1^{1/2} G^T G D_2^{1/2} x_2\\
    &\leq \norm{x_1}{}^2 \norm{D_1^{1/2}G^T G D_1^{1/2}}{} + \\
    &\hspace{0.5cm}\norm{x_2}{}^2 \norm{D_2^{1/2}G^T G D_2^{1/2}}{} \\
    &\hspace{0.5cm}+ 2 \norm{x_1}{}\norm{x_2}{} \norm{D_1^{1/2}G^T G D_2^{1/2}}{}.
\end{align*}
We bound each of these operator norms in turn by using Theorem~\ref{thm:nelson_woodruff} above.

Note that $D_1$ has support of size at most $4/\lambda_1^2$ since $\norm{D_1}{F}^2\leq 1,$ and so $\Tr(D_1)\leq \frac{4}{\lambda_1}.$ Taking $k=\frac{1}{\lambda_1^2}$, $\eps=\frac{1}{\sqrt{m}\lambda_1}$, and $\delta=\frac{1}{60} 2^{-1/\lambda_1^2}$ in Theorem~\ref{thm:nelson_woodruff} and applying the triangle inequality, we get
\begin{align*}
\norm{D_1^{1/2} G^T G D_1^{1/2}}{} 
&\leq \lambda_1 + \eps\left(\norm{D_1^{1/2}}{}^2 + \frac{\norm{D_1^{1/2}}{F}^2}{k}\right)\\
&\leq \lambda_1 + \eps\left(\lambda_1 + \frac{\Tr(D_1)}{k}\right)\\
&\leq \lambda_1 + \eps\left(\lambda_1 + \frac{4}{\lambda_1 k}\right)\\
&\leq \lambda_1 + \frac{5}{\sqrt{m}}
\end{align*}

Similarly for the second term, we note that $\Tr(D_2)\leq \frac{\lambda_1}{2}n$, and apply Theorem~\ref{thm:nelson_woodruff} with $k=d$, $\eps=1/4$, and $\delta = \frac{1}{60} 2^{-m}$ to get
\begin{align*}
\norm{D_2^{1/2}G^T G D_2^{1/2}}{} 
&\leq \frac{\lambda_1}{2} + \eps\left(\frac{\lambda_1}{2} + \frac{\Tr(D_2)}{k}\right)\\
&\leq \frac{\lambda_1}{2} + \frac{1}{4}\left(\frac{\lambda_1}{2} + \frac{\Tr(D_2)}{d}\right)\\
&\leq \frac{\lambda_1}{2} + \frac{1}{4}\left(\frac{\lambda_1}{2} + \frac{\lambda_1}{2}\right)\\
&= \frac34 \lambda_1.
\end{align*}

For the third term we choose $k=\sqrt{d}/\lambda_1$, $\eps = 1/(\sqrt{\lambda_1} m^{1/4})$, and $\delta=\frac{1}{60} 2^{-\sqrt{m}/\lambda_1}$ which gives
\begin{align*}
    \norm{D_1^{1/2} G^T G D_2^{1/2}}{}
    &\leq \eps \sqrt{\lambda_1 + \frac{\Tr(D_1)}{k}} \sqrt{\frac{\lambda_1}{2} + \frac{\Tr(D_2)}{k}}\\
    &\leq \eps \sqrt{\lambda_1 + \frac{\sqrt{d}}{k}} \sqrt{\frac{\lambda_1}{2} + \frac{\sqrt{d}}{k}}\\
    &\leq \eps \left(\lambda_1 + \frac{\sqrt{d}}{k} \right)\\
    &\leq 2 \frac{\sqrt{\lambda_1}}{m^{1/4}}.
\end{align*}

Note that each application of Theorem~\ref{thm:nelson_woodruff} above allows $G$ to have have $\Theta(m)$ rows provided that $m\geq d.$ Also note that each failure probability above is bounded by $\frac{1}{60} 2^{-\min(m,1/\lambda_1^2)}$, since $\frac{\sqrt{m}}{\lambda_1} \geq \min(m, \frac{1}{\lambda_1^2}).$

Thus we conclude with probability at least $1-\frac{1}{20}2^{-\min(m,1/\lambda_1^2)}$, that
\begin{align*}
    x^T D^{1/2} G^T G D^{1/2} x 
    \leq \left(\lambda_1+\frac{5}{\sqrt{m}}\right)\norm{x_1}{}^2 + \frac{3}{4}\lambda_1\norm{x_2}{}^2 + 4\frac{\sqrt{\lambda_1}}{m^{1/4}}\norm{x_1}{}\norm{x_2}{}.
\end{align*}
We view the right-hand expression as a quadratic form applied to the unit vector $(\norm{x_1}{}, \norm{x_2}{}).$ So its value is bounded by the largest eigenvalue of the $2\times 2$ matrix
\begin{equation*}
M=
\begin{pmatrix}
\lambda_1 + \frac{5}{\sqrt{m}} & \frac{2\sqrt{\lambda_1}}{m^{1/4}} \\
\frac{2\sqrt{\lambda_1}}{m^{1/4}} & \frac34 \lambda_1
\end{pmatrix}.
\end{equation*}
Suppose that $\lambda_1 + \beta$ with $\beta\geq 0$ is an eigenvalue of $M.$  Then plugging into the characteristic polynomial gives
\[
\frac{4\lambda_1}{\sqrt{m}} = \left(\beta - \frac{5}{\sqrt{m}}\right)\left(\beta + \frac{\lambda_1}{4}\right) \geq \frac{\lambda_1}{4}\left(\beta - \frac{5}{\sqrt{m}}\right),
\]
from which it follows that $\beta \leq O\left(\frac{1}{\sqrt{m}}\right)$ as desired.

\end{proof}

\begin{lem}
\label{lem:sketched_eval_upper_bound}
Let $D\in \R^{d\times d}$ (not necessarily PSD) have $\norm{D}{F}\leq 1,$ and suppose $\lambda_{\ell}(D) \geq 0.$  Let $G\in \R^{k\times d}$ have i.i.d. $\mathcal{N}(0,1/k)$ entries. Then with probability at least $1 - \frac{1}{20}2^{-\min(\ell, \eps^{-2})}$,
\[
\lambda_{\ell}(GDG^T) \leq \lambda_{\ell}(D) + O\parens{\eps},
\]
for $k \geq \Omega(d + \frac{1}{\eps^2}).$

\end{lem}

First we have the following, where $D_+$ and $D_{-}$ denote the positive and negative semi-definite parts of $D$:
\begin{align*}
    \lambda_{\ell}(G D G^T) &= \lambda_{\ell}(G D_{+} G^T - GD_{-}G^T)\\
    &\leq \lambda_{\ell}(GD_+ G^T)\\
    &=\lambda_{\ell}(D_+^{1/2} G^T G D_+^{1/2}).
\end{align*}
Let $S_{d - \ell + 1}$ be the span of a set of eigenvectors of $D$ corresponding to $\lambda_{\ell}(D),\ldots, \lambda_d(D).$ Then by Courant-Fischer\footnote{For example see \cite{vershynin2018high} for a statement of the Courant-Fischer minimax theorem.},
\begin{align*}
    \lambda_{\ell}(G D G^T)
    &\leq \max_{v\in S_{d-\ell+1}, \norm{v}{}=1} v^T D_{+}^{1/2} G^T G D_{+}^{1/2}v\\
    &= \max_{v\in S_{d-\ell+1}, \norm{v}{}=1} \norm{GD_{+}^{1/2}v}{}^2\\
    &= \norm{GD_{+,-(\ell-1)}^{1/2}}{}^2,
\end{align*}
where $D_{+,-(\ell-1)}$ is $D_+$ with the top $\ell-1$ eigenvalues zeroed out. Now Lemma~\ref{lem:op_norm_upper_bound} applies, and gives 
\[
\lambda_{\ell}(GDG^T) \leq \lambda_{\ell}(D_+) + O\parens{\eps} = \lambda_{\ell}(D) + O\parens{\eps},
\]
with probability at least $1 - \frac{1}{20}2^{-\min(1/\eps^2,1/\lambda_{\ell}(D)^2)},$ for $k \geq \Omega(d + \frac{1}{\eps^2}).$  Finally, note that $\lambda_{\ell}(D) \leq \frac{1}{\sqrt{\ell}},$ so 
\[2^{-\min(1/\eps^2,1/\lambda_{\ell}(D)^2)} \leq 2^{-\min(1/\eps^2,\ell)}.\]

\subsection{Lower bounds on the sketched eigenvalues}

\begin{lem}
\label{lem:op_norm_over_subspace_upper_bound}
Let $M\in \R^{d\times d}$ be a PSD matrix with $\norm{M}{F}\leq 1.$  Let $G\in \R^{m\times d}$ have i.i.d. $\mathcal{N}(0,\frac1m)$ entries, where $m\geq \Omega(d + \log(1/\delta))$.  Also let $S_{\ell}$ denote an arbitrary $\ell$ dimensional subspace of $\R^m.$  Then with probability at least $1-\delta$, we have
\[ 
\max_{v\in S_{\ell}, \norm{v}{}=1} v^T G M G^T v \leq 3\frac{\ell}{m}\norm{M}{}.
\]
\end{lem}

\begin{proof}
Let $\Pi\in \R^{m\times \ell}$ has columns forming an orthonormal basis of $S_{\ell}.$   Then we can write
\[
\max_{v\in S_{\ell}, \norm{v}{}=1} v^T G M G^T v = \norm{\Pi^T G M G^T \Pi}{}.
\]
Using rotational invariance of $G$ we note that $\Pi^T G$ is distributed as $\sqrt{\frac{\ell}{m}}\tilde{G}$ where $\tilde{G} \in \R^{\ell \times d}$ has i.i.d. $\mathcal{N}(0,\frac{1}{\ell})$ entries. Then
\[
\norm{\Pi^T G M G^T \Pi}{} = \frac{\ell}{m}\norm{\tilde{G}M\tilde{G}^T}{} = \frac{\ell}{m}\norm{M^{1/2}\tilde{G}^T \tilde{G} M^{1/2}}{},
\]
which by taking $(\eps,k) = (1,d)$ in Theorem~\ref{thm:nelson_woodruff} is bounded by
\begin{align*}
\frac{\ell}{m}\parens{\norm{M}{} + \left(\norm{M^{1/2}}{}^2 + \frac{\norm{M^{1/2}}{F}^2}{d}\right) }
&= \frac{\ell}{m}\parens{\norm{M}{} + \left(\norm{M}{} + \frac{\Tr(M)}{d}\right)} \\
&\leq 3\frac{\ell}{m}\norm{M}{},
\end{align*}
with probability at least $1-\delta.$ Note that we used the bound $\Tr(M) \leq d \norm{M}{}$ in the final step.
\end{proof}

\begin{lem}
\label{lem:psd_lower_bound}
Let $M\in\R^{d\times d}$ be PSD with $\norm{M}{F}\leq 1$, and let $G\in\R^{k\times d}$ have i.i.d. $\mathcal{N}(0,\frac1k)$ entries. 

By choosing $k = \Theta(d + \frac{1}{\eps^2})$ the bound
\[
\lambda_{\ell}(GMG^T) \geq \lambda_{\ell}(M) - \eps
\]
holds with probability at least $1 - \frac{1}{40}2^{-\ell}.$
\end{lem}

\begin{proof}
Recall that the non-zero eigenvalues of $GMG^T$ coincide with those of $M^{1/2}G^TGM^{1/2},$  so
\[
\lambda_{\ell}(GMG^T) = \lambda_{\ell}(M^{1/2}G^TGM^{1/2}).
\]
By the Courant-Fischer theorem, there exists an $\ell$ dimensional subspace $S_{\ell}$ of $\R^d$ such that $\norm{M^{1/2}x}{}^2 = x^T M x \geq \lambda_{\ell}(M)$ for all $x\in S_{\ell}.$

Now suppose that $G$ is an $(\frac{\eps}{\lambda_{\ell}}, \ell, \frac{1}{40}2^{-\ell})$-OSE\footnote{An $(\eps, k, \delta)$-OSE refers to an oblivious embedding that has $1\pm \eps$ distortion over any given $k$ dimensional subspace with probability at least $1-\delta.$}, which can be achieved by taking
\[
k = \Theta\left(\frac{\lambda_{\ell}^2}{\eps^2} \left(\ell + \log\frac{10}{2^{-\ell}}\right)\right).
\]
Since $\norm{M}{F}\leq 1,$ we have $\lambda_{\ell}^2 \leq \frac{1}{\ell}$, so in fact  $k = O(1/\eps^2)$ above.

Then with probability at least $1-\frac{1}{10}2^{-\ell}$, the bound 
\begin{align*}
\norm{GM^{1/2}x}{}^2 &\geq \left(1-\frac{\eps}{\lambda_{\ell}(M)}\right)\norm{M^{1/2}x}{}^2 \\
&\geq\left(1-\frac{\eps}{\lambda_{\ell}(M)}\right)\lambda_{\ell}(M)\\
&\geq \lambda_{\ell}(M) - \eps
\end{align*}
holds for all $x\in S_{\ell}$.  By the Courant-Fischer theorem,  this implies that $\lambda_{\ell}(M^{1/2}G^TGM^{1/2}) \geq \lambda_{\ell}(M) - \eps$ as desired.
\end{proof}

\begin{lem}
\label{lem:sketched_eval_lower_bound}
Suppose that $D \in \R^{d\times d}$ is a (not necessarily PSD) matrix with $\norm{D}{F}\leq 1$ and that $G\in \R^{k\times d}$ has i.i.d. $\mathcal{N}(0,1/k)$ entries. If $\lambda_{\ell}(D)\geq 0,$ then with probability at least $\frac{1}{20}2^{-\ell}$,
\[
\lambda_{\ell}(GDG^T) \geq \lambda_{\ell}(D) - \eps,
\]
for $k \geq \Omega(d + \frac{1}{\eps^2}).$

\end{lem}

Throughout the course of this argument we will need the parameters $k$ and $r$ to satisfy various inequalities. To streamline the proof we will list these assumptions here and later verify that they are satisfied with appropriate choices.  The assumptions we will need are as follows:
\begin{enumerate}
    \item $k \geq c_1 d$, where $c_1\geq 1$ is an absolute constant
    \item $k - r \geq \frac{c_2}{\eps^2}$ where $c_2$ is an absolute constant
    \item $\frac{r}{k\sqrt{\ell}} \leq \eps$
    \item $\frac{\ell}{k\sqrt{r}} \leq \eps$
\end{enumerate}

To produce a lower bound on $\lambda_{\ell}(GDG^T)$ we will find a subspace $S$ such that $v^T G D G^T v$ is large for all unit vectors $v$ in $S$.

First  we write $D = D_+ - (D_{-,-r} + D_{-,+r})$ where $D_+$ is the positive semi-definite part of $D$, $D_-$ is the negative semi-definite part of $D$, $D_{-,+r}$ denotes $D_-$ with all but the top $r$ eigenvalues zeroed out, and $D_{-,-r} = D_- - D_{-,+r}$ (recall that $r$ is the parameter from above which is to be chosen later). We also write 
\begin{align*}
    GDG^T &= GD_+G^T - GD_{-,+r}G^T - GD_{-,-r}G^T  \\
    &= G_1D_+G_1^T - G_2 D_{-,+r}G_2^T - G_3D_{-,-r}G_3^T 
\end{align*}
where each component is PSD, and where $G_1,G_2,G_3$ consist of the columns of $G$ corresponding to the nonzero entries of $D_+$ and $D_{-,+r}$ and $D_{-,-r}$ respectively.  In particular note that this decomposition shows that these three random matrices are mutually independent.

Let $W_r \subseteq \R^k$ denote the image of $D_{-,+r}$ so that $W_r^{\perp} = \ker(D_{-,+r}).$ Let $\Pi_{W_r^{\perp}}\in\R^{k\times(k-r)}$ have columns forming an orthonormal basis for $W_r^{\perp}.$ By rotational  invariance of $G$, $G^T \Pi_{W_r^{\perp}}$ has i.i.d. $\mathcal{N}(0,1/k)$ entries. Thus it follows that
\[
\Pi_{W_r^{\perp}}^T G D_+ G^T \Pi_{W_r^{\perp}} 
\sim \frac{k-r}{k} \tilde{G} D_+ \tilde{G}^T
\sim \left(1-\frac{r}{k}\right)\tilde{G} D_+ \tilde{G}^T,
\]
where $\tilde{G} \in \R^{(k-r)\times d}$ has i.i.d $\mathcal{N}(0,\frac{1}{k-r})$ entries.

Now by Lemma~\ref{lem:psd_lower_bound}, along with our second assumption above, we have
\[
\lambda_{\ell}(\tilde{G} D_+ \tilde{G}^T) \geq \lambda_{\ell}(D_+) - \eps
= \lambda_{\ell}(D) - \eps,
\]
with probability at least $1 - \frac{1}{40}2^{-\ell}.$
Thus with the same probability, we then have
\[
\lambda_{\ell}(\Pi_{W_r^{\perp}}^T G D_+ G^T \Pi_{W_r^{\perp}})
\geq \parens{1-\frac{r}{k}}(\lambda_{\ell}(D) - \eps)
\geq \lambda_{\ell}(D) - 2\eps,
\]
where the last inequality follows from our third assumption above, along with the observation that $\lambda_{\ell}(D) \leq \frac{1}{\sqrt{\ell}}$ which comes from the assumption $\norm{D}{F}\leq 1$.

If the above holds, then by the Courant-Fischer theorem, there exists a subspace $S_{\ell}\subseteq W_r^{\perp}\subseteq \R^k$ such that
\begin{equation}
\label{eq:pos_part_upper_bound}
x^T G D_+ G^T x \geq \lambda_{\ell}(D) - 2\eps
\end{equation}
for all $x\in S_{\ell}.$  Note that the construction of $S_{\ell}$ was independent of $GD_{-,-r}G^T$ by the comment above.  Thus we may apply Lemma~\ref{lem:op_norm_over_subspace_upper_bound}, along with our first assumption, to conclude that with probability at least $1-\frac{1}{40}2^{-d}$,
\begin{equation}
\label{eq:neg_part_upper_bound}
\max_{v\in S_{\ell},\norm{v}{}=1} v^T G D_{-,-r}G^T v \leq 3\frac{\ell}{k}\norm{D_{-,-r}}{} \leq 3\frac{\ell}{k}\frac{1}{\sqrt{r}}.
\end{equation}
The last inequality holds because $\norm{D_-}{F}=1$, which implies that $\lambda_r(D_-) \leq \frac{1}{\sqrt{r}}.$

Now let $u\in S_{\ell}$ be an arbitrary unit vector. We write
\[
uGDG^Tu^T = u^T G D_+ G^T u - u^T G D_{-,-r} G^T u - u^T G D_{-,+r} G^T u.
\]
The last term vanishes by design since $x\in W_r^{\perp}.$  We then bound the first term using equation~\ref{eq:pos_part_upper_bound} and the second term using equation~\ref{eq:neg_part_upper_bound} to get
\[
uGDG^Tu^T \geq (\lambda_{\ell}(D) - 2\eps) - 3\frac{\ell}{k}\frac{1}{\sqrt{r}} \geq \lambda_{\ell}(D) - 5\eps,
\]
where the second inequality is form the fourth assumption above.

Our total failure probability in the argument above is at most $\frac{1}{40} 2^{-d} + \frac{1}{40}2^{-\ell} \leq \frac{1}{20}2^{-\ell}$ as desired. It remains to choose parameters so that our four assumptions are satisfied.  For this we take
\begin{align*}
k &\geq \max\left(c_1 d, \frac{c_2}{\eps^2} + \lfloor 2\ell\rfloor, \frac{2\sqrt{\ell}}{\eps}\right)\\
r &= \lfloor 2\ell \rfloor.
\end{align*}
Assumptions 1 and 2 clearly hold with this choice. For assumption 3, we have
\[\eps k \sqrt{\ell} \geq \eps \frac{2\sqrt{\ell}}{\eps}\sqrt{\ell} = 2\ell \geq r,\]
and for assumption 4,
\[
\eps k \sqrt{r} \geq \eps \frac{2\sqrt{\ell}}{\eps}\sqrt{2\ell-1} = 2\sqrt{\ell}\sqrt{2\ell-1} \geq \ell,
\]
since $\ell \geq 1.$  Finally, since $\ell\leq d,$ this gives a bound of $k = O(d + \frac{1}{\eps^2})$ as desired (note the inequality $\frac{\sqrt{d}}{\eps} \leq \max(d,1/\eps^2)$ for bounding the last term in the $\max$ defining $k$).

\subsection{Controlling the Tail}
In this section we use Hanson-Wright\footnote{See \cite{vershynin2018high} for a precise statement of Hanson-Wright.} to bound the effect of the tail eigenvalues of $A$ on the sketch.  Note that our application Hanson-Wright relies on Gaussianity of $G$ in order for the entries of $G^T u$ to be independent. 
\begin{lem}
\label{lem:tail_bound}
Let $Y\in \R^{d\times d}$ be symmetric (not necessarily PSD) with $\norm{Y}{}\leq \eps$ and $\norm{Y}{F}\leq 1$ . Let $G\in \R^{k\times n}$ have i.i.d. $\mathcal{N}(0,1/k)$ entries.  For $k\geq \Omega(1/\eps^2)$ we have
\[
\norm{GYG^T - 
\frac{1}{k}\Tr(Y)I}{} \leq O(\eps),
\]
with probability at least $29/30.$
\end{lem}

\begin{proof}
Let $u\in \R^k$ be an arbitrary fixed unit vector.  Note that $G^T u$ is distributed as $\mathcal{N}(0,\frac{1}{k}I_d)$ and so 
\[
\E(u^T G Y G^T u) = \frac{1}{k}\Tr(Y). 
\]

Set $\tilde{Y} = GYG^T - \frac{\Tr(Y)}{k}I.$ By Hanson-Wright,
\begin{align*}
\Pr\left(\left|u^T \tilde{Y} u\right| \geq 30\eps \right)
&=\Pr\left(\left|u^T G Y G^T u - \frac{1}{k}\Tr(Y)\right| \geq 30\eps \right)\\
&\leq 2\exp\left( -0.1 \min\left(\frac{(30\eps)^2 k^2}{\norm{Y}{F}^2} , \frac{(30\eps) k}{\norm{Y}{2}}  \right) \right)\\
&\leq 2\exp\left( - \min\left(90 \eps^2 k^2 , 3k  \right) \right).
\end{align*}
Note that in the final bound above we used the fact that $\norm{Y}{2}\leq \eps$.

Let $\mathcal{N}$ be a net for the sphere in $\R^k$ with mesh size $1/3,$ which may be taken to have size $9^k.$  By 4.4.3 in \cite{vershynin2018high},
\[\norm{G\tilde{Y}G^T}{2} \leq 3\sup_{x\in \mathcal{N}} |x^T G\tilde{Y}G^T x|.\]
By taking a union bound over the net and setting $k\geq \Omega(1/\eps^2)$, we then have
\[\Pr\left(\norm{\tilde{Y}}{2} \geq 93\eps\right) \leq 2\exp\left( - \min\left(90 \eps^2 k^2 , 3k  \right) \right) 9^k \leq \frac{1}{30},\]
for $\eps<1$.
\end{proof}




\subsection{Proof of Theorem~\ref{thm:spectral_apx}}

\begin{proof}
By rescaling, it suffices to consider that case $\norm{A}{F}=1.$ We start by decomposing $A$ into two pieces $A = A_1 + A_2$, where $A_1$ is $A$ with all eigenvalues smaller than $\eps$ in magnitude zeroed out.

To handle the large eigenvalues, we apply Lemma~\ref{lem:sketched_eval_upper_bound} and Lemma~\ref{lem:sketched_eval_lower_bound}.  Suppose that $A_1$ has $n$ nonzero eigenvalues.  Then we note that the nonzero eigenvalues of $GA_1G^T$ have the same distribution as the eigenvalues of $\tilde{G} \tilde{A}_1 \tilde{G}^T$ where $\tilde{A}_1$ is a symmetric $n\times n$ matrix with eigenvalues the same as the nonzero eigenvalues of $A_1$ and where $\tilde{G} \in \R^{k\times n}$ has i.i.d. $\mathcal{N}(0,1/k)$ entries.  This effectively means that we may treat $A_1$ has having dimension $n$ when applying Lemma~\ref{lem:sketched_eval_upper_bound} and Lemma~\ref{lem:sketched_eval_lower_bound}.

By taking a union bound over the positive eigenvalues of $A_1$ and applying Lemma~\ref{lem:sketched_eval_upper_bound} we get the upper bound $\lambda_{\ell}(GA_1 G^T) \leq \lambda_{\ell}(A_1) + O(\eps)$ uniformly for all $\ell$ such that $\lambda_{\ell}(A_1)> 0$, with failure probability at most
\[
\sum_{i=1}^n \frac{1}{20} 2^{-\min(\ell, \eps^{-2})} \leq \frac{1}{20}\sum_{i=1}^n 2^{-\ell} \leq \frac{1}{20},
\]
where the the first inequality follows from the fact that $\ell \leq n \leq 1/\eps^2$, which in turn holds since $\norm{A_1}{F}\leq 1.$

Similarly Lemma~\ref{lem:sketched_eval_lower_bound} gives the lower bound $\lambda_{\ell}(G A_1 G^T) \leq \lambda_{\ell}(A_1) - \eps$ uniformly for all $\ell$ such that $\lambda_{\ell}(A_1)> 0$, with failure probability at most
\[
\sum_{i=1}^{\ell}\frac{1}{20}2^{-\ell} \leq \frac{1}{20}.
\]
Thus with at least $9/10$ probability, $\abs{\lambda_{\ell}(GA_1G^T) - \lambda_{\ell}(A_1)} \leq O(\eps)$ for all $\ell$ such that $\lambda_{\ell}(A_1) > 0.$  By applying the above argument to $-A_1$ we get the same guarantee for the negative eigenvalues, i.e. $\abs{\lambda_{k-\ell}(GA_1G^T) - \lambda_{k-\ell}(A_1)} \leq O(\eps)$ for all $\ell$ such that $\lambda_{k-\ell}(A_1) < 0.$ By a union bound, the positive and negative guarantees hold together with failure probability at most $1/5$.

Next we apply the tail bound of Lemma~\ref{lem:tail_bound} to control the perturbations resulting from the tail. By the triangle inequality,
\begin{align*}
\norm{GA_2G^T - 
\frac{1}{k}\Tr(GAG^T)I}{} 
&\leq \norm{GA_2G^T - \frac1k\Tr(A_2)I}{} \\
&\hspace{0.5cm}+ \norm{\frac1k\Tr(A_2)I - \frac{1}{k}\Tr(GAG^T)I}{}\\
&\leq \norm{GA_2G^T - \frac1k\Tr(A_2)I}{}\\
&\hspace{0.5cm}+ \frac1k\abs{\Tr(A_2) - \Tr(GA_2G^T)}\\
&\hspace{0.5cm}+ \frac1k\abs{\Tr(GA_1G^T)}
\end{align*}
The first of these terms is bounded by $O(\eps)$ with failure probability at most $1/30$ by Lemma~\ref{lem:tail_bound}. The second term is easily bounded by $O(\eps)$ with failure probability at most $1/30$ since $\Tr(GA_2G^T)$ is a trace estimator for $A_2$ with variance at $O(\norm{A_2}{F})=O(1)$ (in fact the variance is even smaller).  For the third term, note that $A_1$ has at most $1/\eps^2$ nonzero eigenvalues, so $\Tr(A_1) \leq \frac{1}{\eps}\norm{A}{F} \leq \frac{1}{\eps}.$ Thus since $\Tr(GA_1G^T)$ is a trace estimator for $A_1$, the third term is bounded by $O(\eps)$ with failure probability at most $1/30$.  Thus we have the bound
\[
\norm{GA_2G^T - 
\frac{1}{k}\Tr(GAG^T)I}{}  \leq O(\eps),
\]
with failure probability at most $1/10.$  This gives the bound
\begin{align*}
    \lambda_{\ell}(GAG^T) &= \lambda_{\ell}(GA_1G^T + GA_2G^T)\\
    &= \lambda_{\ell}\left(GA_1G^T + \frac1k \Tr(GAG^T)I + GA_2G^T -\frac1k \Tr(GAG^T)I \right)\\
    &= \lambda_{\ell}\left(GA_1G^T +\frac1k \Tr(GAG^T)I\right) \\
    &\hspace{0.5cm} \pm \norm{GA_2G^T -\frac1k\Tr(GAG^T)I}{2}\\
    &= \lambda_{\ell}(GA_1G^T) + \frac1k\Tr(GAG^T) \pm O(\eps).
\end{align*}
Setting $\widehat{\lambda_{\ell}} = \lambda_{\ell}(GAG^T) - \frac1k \Tr(GAG^T),$ we therefore have $\widehat{\lambda_{\ell}} = \lambda_{\ell}(GA_1G^T) \pm O(\eps).$  Combining with the bounds above gives $\widehat{\lambda_{\ell}} = \lambda_{\ell}(A_1) \pm O(\eps)$ if $\lambda_{\ell}(A_1) > 0$ and $\widehat{\lambda_{k-\ell}} = \lambda_{k-\ell}(A_1) \pm O(\eps)$ if $\lambda_{k-\ell}(A_1) > 0.$

Thus there is a subset of $n$ of the $\widehat{\lambda_{\ell}}$'s which provide an $O(\eps)$ additive approximation to the set of eigenvalues of $A$ which are at least $\eps.$  The above bound shows that the remaining $\widehat{\lambda_{\ell}}$'s are bounded by $O(\eps)$ and the result follows.
\end{proof}

\section{Lower bounds for eigenvalue estimation}

We will use the Wishart distribution throughout this section which is defined as follows.

\begin{definition}
The $n$ dimensional Wishart distribution with $r$ degrees of freedom $W(n,r)$ is the distribution of $GG^T$ where $G\in \R^{n\times r}$ has i.i.d. standard normal entries.
\end{definition}

In this section we show that $\Omega(r)$ matrix-vector queries are necessary to determine the rank of a matrix with all nonzero entries $\Omega(1).$  Specifically we show that distinguishing between $W(n,r)$ and $W(n,r+2)$ requires $\Omega(r)$ queries for $r \leq O(n).$ In Appendix~\ref{appendix:lower_bound_projection} we sketch a proof of a similar lower bound for determining the rank of the orthogonal projection onto a random subspace.

For now we consider the following problem.

\begin{problem}
\label{prob:distinguishing_wisharts}
Given a matrix $A$ sampled from either $\mathcal{D}_1 = W(n,r)$ or $\mathcal{D}_2 = W(n,r+2)$ each with equal probability, decide between $\mathcal{D}_1$ and $\mathcal{D}_2$ with at least $2/3$ probability, using (possibly adaptive) matrix-vector queries to $A$.
\end{problem}

We first make note of the following result, which is effectively a version of Lemma 13 from \cite{braverman2020gradient}, adapted to Wishart matrices $W(n,r)$ with $n$ and $r$ not necessarily equal.  This will allow us to show that adaptivity is unhelpful, and hence reduce to studying the non-adaptive case.

\begin{prop}
\label{prop:conditional_distribution}
Let $A\sim W(n,r),$ and let $k<r\leq n.$  Then the conditional distribution $A|\{Ae_1 = x_1, \ldots, Ae_k = x_k\}$ can be written as
\[
M_k + \diag(0_{k\times k}, W(n-k, r-k)),
\]
where $M_k \in \R^{n\times n}$ has rank at most $k$ and depends only on $x_1,\ldots, x_k$. In particular $M_k$ does not depend on $r.$
\end{prop}

\begin{proof}
Write $A = GG^T$ where $G\in \R^{n\times r}$ has i.i.d. $\mathcal{N}(0,1)$ entries. Write $g_1,g_2,\ldots$ for the rows of $G$.  We first consider the conditional distribution $A|\{Ae_1 = x_1\}.$  In other words, we are conditioning on the events $\inner{g_1}{g_i} = x_{1i}$ for all $i$.  By rotational invariance, we may additionally condition on $g_1 = \sqrt{x_{11}}e_1$ without changing the resulting distribution.  Then for $i > 1$, the conditional distribution of $g_i$ can be written as $\frac{x_{1i}}{\sqrt{x_{11}}} e_1 + h_i$ where $h_i$ is distributed as $\mathcal{N}(0, I_{n-1})$ in the orthogonal complement of $e_1.$  It follows from this that we can write
\begin{equation}
\label{eqn:first_step_wishart_induction}
A|\{Ae_1 = x_1\} \sim \frac{1}{x_{11}} x_1x_1^T + \diag(0,W(n-1,r-1)).
\end{equation}
So we have $M_1 = \frac{1}{x_{11}} x_1x_1^T$. Now we apply the above line inductively. 

For $j < r$, let $W_j\sim \diag(0_{k\times k}, W(n-j, r-j)),$ and write
\begin{align*}
A | \{Ae_1 = x_1, \ldots Ae_{j+1} = x_j\}  &\sim \left(A | \{Ae_1 = x_1, \ldots Ae_j = x_j\} \right) | \{A e_{j+1} = x_{j+1}\}\\
&\sim (M_j + W_j)|\{(M_j + W_j) e_{j+1} = x_{j+1}\}\\
&\sim (M_j + W_j) | \{W_j e_{j+1} = x_{j+1} - M_j e_{j+1}\}\\
&\sim (M_j + W_j) | \{W_j e_{j+1} = v_{j+1}\}\\
&\sim M_j + \left(W_j | \{W_j e_{j+1} = v_{j+1}\}\right)
\end{align*}
where we set $v_{j+1} = x_{j+1} - M_j e_{j+1}.$

By applying \ref{eqn:first_step_wishart_induction},
\[
\{W_j e_{j+1} = v_{j+1}\} = \frac{1}{v_{j+1,j+1}} v_{j+1}v_{j+1}^T + W_{j+1}.
\]

Hence we can take
\[
M_{j+1} = M_j + \frac{1}{v_{j+1,j+1}} v_{j+1}v_{j+1}^T,
\]
and the induction is complete.

\end{proof}


\begin{prop}
\label{prop:wishart_adaptivity_unhelpful}
Of all (possibly adaptive) algorithms for Problem~\ref{prob:distinguishing_wisharts} which make $k\leq r$ queries, there is an optimal such algorithm (in the sense of minimizing the failure probability), which queries on the standard basis vectors $e_1, \ldots, e_k$.
\end{prop}

\begin{proof}
Let $s$ be either $r$ or $r+2$ corresponding to which of $\mathcal{D}_1$ and $\mathcal{D}_2$ is sampled from. By rescaling, we assume that only unit vectors are queried.

We argue by induction. Since $\mathcal{D}_1$ and $\mathcal{D}_2$ are rotationally invariant, we may without loss of generality take the first query to be $e_1.$

Now suppose inductively that there is an optimal $k$ query algorithm $\mathcal{A}$ whose first $j$ queries are always $e_1, \ldots, e_j.$  Suppose on a fixed run, that $Ae_1 = x_1, \ldots, Ae_j = x_j.$  By Proposition~\ref{prop:conditional_distribution}, we may write the resulting conditional distribution as
\[
A|\{Ae_1 = x_1, \ldots Ae_j = x_j\} = M_j + A_j,
\]
where $M_j$ depends deterministically on $x_1, \ldots, x_j$ (and not on $s$), and $A_j \sim \diag(0_{j\times j}, W(n-j, s-j))$.

Now since $M_j$ is know to $\mathcal{A}$, we may assume that on iteration $j+1$, $\mathcal{A}$ is given matrix-vector query access to $A_j,$ rather than to $A$.  Since the first $j$ rows and columns of $A_j$ are filled with zeros, we may assume that $\mathcal{A}$ queries on a vector in $\spn\{e_{j+1}, \ldots, e_n\}.$ Then by rotational invariance of $W(n-j, s-j)$, we may take $\mathcal{A}$ to query on $e_j$ on iteration $j+1.$ This completes the induction, and the claim follows.
\end{proof}

In light of the previous result, only non-adaptive queries are necessary.  In fact we can make an even stronger claim.  Let $E_k$ denote the matrix with columns $e_1,\ldots, e_k.$  The previous proposition showed that an optimal tester only needs to observe $A E_k$, the first $k$ columns of $A.$  In fact, only $E_k^T A E_k$, the leading principal submatrix of $A$ is relevant. We first state a simple fact that drives the argument.

\begin{prop}
\label{prop:simple_rotation_invariance_consequence}
Let $X\in k\times r_1$ and $Y\in k\times r_2$ be fixed matrices such that $XX^T = YY^T.$ Let $v_1\in \R^{r_1}$ and $v_2 \in \R^{r_2}$ have i.i.d. standard normal entries.  Then $Xv_1$ and $Yv_2$ have the same distribution.
\end{prop}

\begin{proof}

Suppose without loss of generality that $r_2\geq r_1.$ Then since $XX^T = YY^T$, there is an orthogonal matrix $U \in \R^{r_2\times r_2}$ such that
\[
YU = [X, 0_{k\times (r_1-r_2)}].
\]
Now let $g\in \R^{r_2}$ have i.i.d. standard normal entries.  By rotational invariance $Ug \in \R^{r_2}$ does as well.  So $YU$ has the same distribution as $Yv_2.$  Also $[X, 0_{k\times (r_1-r_2)}] g$ is distributed as $Xv_1$, so $Xv_1$ and $Yv_2$ have the same distribution as desired.
\end{proof}

\begin{prop}
\label{prop:top_corner_suffices}
Suppose that $A_1 \sim W(n,r)$ and $A_2 \sim W(n,r+2).$  Then for $k\leq r$, \[\tv(A_1 E_k, A_2 E_k) = \tv(E_k^T A_1 E_k, E_k^T A_2 E_k).\]
\end{prop}

\begin{proof}
Let $G_1 \in \R^{k\times r}$ and $H_1 \in \R^{(n-k) \times r}$ have i.i.d. standard normal entries.  Similarly let $G_2 \in \R^{k\times (r+2)}$ and $H_2 \in \R^{(n-k) \times (r+2)}$ have i.i.d. standard normal entries. 

By the definition of the Wishart distribution, the joint distribution of the entries of $A_1 E_k$ is precisely that of $(G_1 G_1^T, H_1 G_1^T)$ and similarly for $A_2 E_k.$ Hence,
\[
\tv(A_1 E_k, A_2 E_k) = \tv\left((G_1 G_1^T, H_1 G_1^T), (G_2 G_2^T, H_2 G_2^T)\right).
\]

For a fixed matrix $M$ of the appropriate dimensions, we consider the conditional distribution $H_i G_i^T | \{G_i G_i^T = M\}$ for $i=1,2.$  The rows of this random matrix are independent (since the rows of $H_i$ are independent), and by Proposition~\ref{prop:simple_rotation_invariance_consequence} the distribution of each row is a function of $M$.  Hence it follows that
\[
H_1 G_1^T | \{G_1 G_1^T = M\} = H_2 G_2^T | \{G_2 G_2^T = M\}
\]
for all $M$.  Therefore,
\[
\tv\left((G_1 G_1^T, H_1 G_1^T), (G_2 G_2^T, H_2 G_2^T)\right) = \tv(G_1G_1^T, G_2 G_2^T).
\]
Since $E_k^T A_i E_k$ has the same distribution as $G_i G_i^T$, the claim follows.
\end{proof}

Our problem is now reduced to that of determining the degrees of freedom of a Wishart from observing the top corner (which is itself Wishart).  We will give a lower bound for this problem.

Our proof uses the following version of Theorem 5.1 in \cite{jonsson1982some}.
\begin{thm}
\label{thm:wishart_det_convergence}
Let $\alpha\in (0,1)$ be a constant, and let $n,r \rightarrow \infty$ simultaneously, with $n/r \rightarrow \alpha.$  Then
\[
\frac{\det( W(n,r))}{(r-1)(r-2)\ldots(r-n)} \rightarrow e^{\mathcal{N}(0, -2\log(1-\alpha))},
\]
where the convergence is in distribution.
\end{thm}

\begin{lem}
\label{lem:wisharts_close_in_tvd}
Let $\alpha = 0.1$.  There exists a constant $c$ so that if $r \geq c$, then
\[
\tv\left(W(\lfloor \alpha r \rfloor, r), W(\lfloor \alpha r \rfloor, r+2)\right) \leq 0.2.
\]
\end{lem}

\begin{proof}
We write $n=\lfloor \alpha r \rfloor$ with the understanding that $n$ is a function of $r$. Let $\mu_{n,r}$ be the measure on $\R^{n(n+1)/2}$ associated to $W(n,r)$, and let $f_{n,r}$ be the corresponding density function (with respect to the Lebesgue measure). Also let $\Delta_{+} \subseteq \R^{n(n+1)/2}$ be the PSD cone.  Then we have

\begin{align*}
    \tv(W(n,r), W(n,r+2)) 
    &= \int_{\Delta_+} \left (f_{n,r}(A) - f_{n,r+2}(A)\right)_+ d\lambda\\
    &= \int_{\Delta_+} \left (1 - \frac{f_{n,r+2}(A)}{f_{n,r}(A)}\right)_+ d\mu_{n,r}\\
\end{align*}

We recall the following standard formula for the density of the Wishart distribution (see \cite{anderson1962introduction} for example):
\[
f_{n,r}(A) =
\frac{(\det A)^{\frac{1}{2}(r-n-1)}e^{-\frac{1}{2}\Tr(A)}}
{\sqrt{2}^{rn}\pi^{\frac14 n(n-1)}\displaystyle\prod_{i=1}^n \Gamma\left(\frac12(r+1-i)\right)}.
\]

Cancelling and applying the identity $\Gamma(x+1) = x\Gamma(x)$ gives

\begin{align*}
\frac{f_{n,r+2}(A)}{f_{n,r}(A)} 
&= \frac{\det A}{2^n} \prod_{i=1}^n\frac{\Gamma\left(\frac12(r+1-i)\right)}{\Gamma\left(1 + \frac12(r+1-i)\right)}\\
&= \frac{\det A}{2^n}\prod_{i=1}^n \frac{1}{\frac12(r+1-i)}\\
&= \frac{\det A}{r(r-1)\ldots (r-n+1)}.
\end{align*}

This gives 

\begin{align*}
\tv(W(n,r), &W(n,r+2)) 
=\\ 
&\int_{\Delta_+} \left( 1 - \frac{\det A}{r(r-1)\ldots (r-n+1)}\right)_+ d\mu_{n,r}(A)\\
&= \E_{A\sim W(n,r)}\left( 1 - \frac{\det A}{r(r-1)\ldots (r-n+1)}\right)_+.
\end{align*}
Therefore it suffices to bound this expectation.

Since $\frac{r-n}{r} \rightarrow (1-\alpha)$ as $r\rightarrow \infty$ we have from Theorem~\ref{thm:wishart_det_convergence} that
\[
\frac{\det W(n,r)}{r(r-1)\ldots (r-n+1)} \rightarrow (1-\alpha)e^{\mathcal{N}(0,-2\log(1-\alpha))}.
\]
Therefore
\[
\tv(W(n,r), W(n,r+2)) \rightarrow \E_{x\sim \mathcal{N}(0, -2\log(1-\alpha))}\left[1 - (1-\alpha)e^x \right]_+,
\]
where swapping the limit with the expectation was justified since the random variables in the limit were all bounded by $1.$  This last expectation may be computed numerically to be approximately $0.1815$ and the claim follows.
\end{proof}

\begin{thm}
\label{thm:wishart_main_lower_bound}
Suppose that $r\geq C_1$ and $d\geq C_2 r$ for absolute constants $C_1$ and $C_2.$ Let $\mathcal{A}$ be an adaptive algorithm making $k$ matrix-vector queries, which correctly decides between $\mathcal{D}_1$ and $\mathcal{D}_2$ with $2/3$ probability. Then $k\geq r/10.$
\end{thm}

\begin{proof}
Consider a protocol which makes $k$ matrix-vector queries. By Proposition~\ref{prop:wishart_adaptivity_unhelpful} and Proposition~\ref{prop:top_corner_suffices} it suffices to consider non-adaptive protocols which observe $E_k^T \Pi E_k$.  Suppose that $A$ is either drawn from $\mathcal{D}_1$ or $\mathcal{D}_2$ and hence distributed as $W(k,r)$ or $W(k,r+2)$.  Lemma~\ref{lem:wisharts_close_in_tvd} now implies that distinguishing these distributions requires $k\geq r/10$ as desired.


\end{proof}



\begin{corollary}
An algorithm which estimates all eigenvalues of any matrix $A$ up to $\eps\norm{A}{F}$ error, with $3/4$ probability must make at least $\Omega(1/\eps^2)$ matrix-vector queries.
\end{corollary}

\begin{proof}
The nonzero eigenvalues of $W(n,r)$ are precisely the squared singular values of an $n\times r$ matrix with i.i.d. Gaussian entries.  So by standard bounds (see \cite{vershynin2018high} for example), the nonzero eigenvalues of $W(n,r)$ and $W(n,r+2)$ are bounded between $\frac{1}{2}n$ and $2n$ with high probability as long as $n\geq Cr$ for an absolute constant $C$.  Since $W(n,r)$ has rank $r$, the Frobenius norm of $W(n,r)$ is bounded by $2n\sqrt{r},$ and similarly for $W(n,r+2).$  Thus setting $\alpha = \frac{1}{10\sqrt{r+2}},$ we see that an algorithm which estimates all eigenvalues of a matrix to $\alpha \norm{A}{F}$ additive error could distinguish $W(n,r)$ from $W(n,r+2),$ and hence by Theorem~\ref{thm:wishart_main_lower_bound} must make at least $r/10$ queries.  The result follows by setting $r = \Theta(1/\eps^2).$
\end{proof}

\section{Acknowledgements}
D. Woodruff would like to acknowledge partial support from a Simons Investigator Award. W. Swartworth was partially supported by  NSF DMS \#2011140.

The authors would also like to acknowledge Cameron Musco, Deanna Needell, and Gregory Dexter for helpful conversations when preparing this manuscript.

\printbibliography

\appendix
\section{Rank estimation lower bound from random projections}
\label{appendix:lower_bound_projection}

In this section, we show a lower bound on determining the rank of a random orthogonal projection from matrix-vector queries.  The key intuition is that running a power-method type algorithm is unhelpful since projections are idempotent.  This suggests that adaptivity should be unhelpful, and indeed this is the case.

Throughout this section, we let $\mathcal{D}_1 = \mathcal{D}_1(d,r)$ be an orthogonal projection $\R^d\rightarrow \R^d$ onto a random $r$ dimensional subspace (sampled from the rotationally invariant measure), and let $\mathcal{D}_2$ be an orthogonal projection onto a random $r+2$ dimensional subspace.  Let $\mathcal{D}$ be the distribution obtained by sampling from either $\mathcal{D}_1$ or $\mathcal{D}_2$ each with probability $1/2$.  

We first show that adaptivity is unhelpful in distinguishing $\mathcal{D}_1$ from $\mathcal{D}_2.$  To prove this, we first make a simple observation.

\begin{observation}
Suppose that $\mathcal{P}_1$ and $\mathcal{P}_2$ are any distributions over matrices, and let $U$ be an orthogonal matrix.  Suppose that $x_1$ is an optimal first query to distinguish $\mathcal{P}_1$ and $\mathcal{P}_2$. Then $Ux_1$ is an optimal first query to distinguish $U\mathcal{P}_1 U^T$ and $U\mathcal{P}_2 U^T.$
\end{observation}

\begin{lem}
\label{lem:adaptivity_unhelpful}
Suppose that there is a (possibly randomized) adaptive algorithm $\mathcal{A}$ which makes $k$ matrix-vector queries to an orthogonal matrix $\Pi\sim \mathcal{D}$ and then decides whether $\Pi$ was drawn from $\mathcal{D}_1$ or $\mathcal{D}_2$ with advantage $\beta.$  Then there is a non-adaptive algorithm which queries on $e_1,\ldots, e_k$ and also achieves advantage $\beta.$
\end{lem}

\begin{proof}

By Yao's principle, it suffices to consider deterministic protocols, so we will restrict ourselves to deterministic protocols in what follow.

First, let us say that an adaptive protocol making queries $v_1,v_2,\ldots$ is \textit{normalized} if for each $i$, $v_{i+1}$ is in the orthogonal complement of $\spn(v_1,v_2,\ldots, v_i, \Pi v_1,\ldots \Pi v_i),$ and $v_i\neq 0$. We will argue that all normalized protocols making $k$ queries achieve the same advantage.

We first observe that all choices of $v_1$ are equivalent, which is a consequence of rotational invariance along with the observation above.

Suppose that a normalized algorithm makes queries $v_1,\ldots, v_j$ and receives values $y_1,\ldots, y_j$ in the first $j$ rounds.  We observe that the conditional distribution of $\Pi$ under these observations is invariant under the group of orthogonal transformations stabilizing $x_1,\ldots, x_j, y_1,\ldots, y_j.$ Applying the observation to this conditional distribution, again shows that all $x_{j+1}$ are equivalent since the stabilizer of $x_1,\ldots, x_j, y_1,\ldots, y_j$ acts transitively on their orthogonal complement.

Finally we observe that a non-adaptive algorithm which queries on $e_1, \ldots, e_k$ can almost surely simulate a normalized protocol.  Indeed let $P_j$ denote projection onto $\spn(e_1,\ldots, e_j, \Pi e_1, \ldots \Pi e_j).$ Then $e_1, P_1 e_2, \ldots, P_{k-1}e_k$ is almost surely a normalized protocol.  Moreover $\Pi P_{j-1}e_j$ may be computed for each $j$, since the values of $\Pi e_1, \ldots \Pi e_j, \Pi^2 e_1, \Pi^2 e_j$ are all known (this uses that $\Pi$ is a projection and hence idempotent).

\end{proof}

We are now able to turn our attention to non-adaptive algorithms.  Let $E_k \in \R^{d\times k}$ denote the matrix $[e_1,\ldots, e_k].$  As we saw above a general matrix-vector query algorithm might as well observe $\Pi E_k.$  As in our argument for Wishart matrices, our next observation is that only the top $k\times k$ corner is useful.

\begin{lem}
\label{lem:only_top_corner_useful}
Suppose that $\Pi_1\sim \mathcal{D}_1$ and $\Pi_2 \sim \mathcal{D}_2.$ We have that
\[
\tv(\Pi_1 E_k, \Pi_2 E_k) = \tv(E_k^T \Pi_1 E_k, E_k^T \Pi_2 E_k).
\]
\end{lem}

\begin{proof}
Let $\Pi E_k = [M_1; M_2]$ where $M_1 \in \R^{k\times k}$ and $M_2 \in \R^{(d-k) \times k}.$  Observe that since $\Pi$ is a projection, $M_2^T M_2 = M_1 - M_1 M_1^T.$

Let the orthogonal group $SO(n)$ act on $\Pi$ via conjugation.  Let $H$ be the stabilizer of $M_1$ under the action, i.e., the set of $U$ such that $U^T \Pi U E_k = [M_1, M_2']$ for some $M_2'.$  We claim that the orbit of $M_2$ under $H$ is $\{X: X^T X = M_1 - M_1 M_1^T\}.$  To see this, simply observe that $H$ is contained in the stabilizer of $e_1,\ldots, e_k$, which is isomorphic copy of $SO(n-k)$ acting on $\spn(e_1,\ldots, e_k)^{\perp}.$  This latter group acts transitively on $\{X: X^T X = M_1 - M_1 M_1^T\}$ under left multiplication as desired.

This implies that the conditional distribution of $M_2$ on observing $M_1$ is uniform over $\{X: X^T X = M_1 - M_1 M_1^T\}$.  Since the conditional distribution is independent of $r$, the result follows.
\end{proof}

Next we leverage a known result showing that a small principal minor of a random rotation is indistinguishable from Gaussian.  This allows to observe that $E_k^T \Pi E_k$ is nearly indistinguishable from a Wishart distribution when $d$ is large.
\begin{lem}
\label{lem:random_projection_corner_almost_wishart}
Suppose that $r\geq C_1$ and $d\geq C_2 r^2$ for some absolute constants $C_1,C_2$, and let $\Pi \sim \mathcal{D}_1(d,r)$ with $k\leq r.$  Then 
\[
\tv(E_k^T \Pi E_k, W(k,r)) \leq 0.1.
\]
\end{lem}

\begin{proof}
Note that $\Pi$ can be written as $(UE_r)(UE_r)^T$ where $U$ is a random orthogonal matrix sampled according to the Haar measure.  Let $G \in \R^{k\times r}$ be a matrix with i.i.d. $\mathcal{N}(0,\frac1d)$ entries.  Then we have
\begin{align*}
    \tv(E_k^T \Pi E_k, G^T G)
    &= \tv(E_k^T (UE_r)(UE_r)^T E_k, G^T G)\\
    &= \tv((E_k^T U E_r)(E_k^T U E_r)^T, G^T G)\\
    &\leq \tv(E_k^T U E_r, G^T),
\end{align*}
where the last line follows from the data processing inequality.

Note that $E_k^T U E_r$ is simply the top $k\times r$ corner of a random orthogonal matrix, and $G^T$ is a $k\times r$ matrix with i.i.d. $\mathcal{N}(0,\frac1d)$ entries.  The claim now follows from Theorem 1 of \cite{jiang2006many}.
\end{proof}

\begin{thm}
Suppose that $r\geq C_1$ and $d\geq C_2 r^2$ for absolute constants $C_1$ and $C_2.$ Let $\mathcal{A}$ be an adaptive algorithm making $k$ matrix-vector queries to a sample from $\mathcal{D}$ which correctly decides between $\mathcal{D}_1$ and $\mathcal{D}_2$ with $3/4$ probability. Then $k\geq r/10.$
\end{thm}

\begin{proof}
Consider a protocol which makes $k$ matrix-vector queries. By Lemma~\ref{lem:adaptivity_unhelpful} and Lemma~\ref{lem:only_top_corner_useful} it suffices to consider non-adaptive protocols which observe $E_k^T \Pi E_k$.  Suppose that $\Pi_1$ and $\Pi_2$ are random projections drawn from $\mathcal{D}_1$ and $\mathcal{D}_2$ respectively.  Then by Lemma~\ref{lem:random_projection_corner_almost_wishart}, we have 
\[\tv(E_k^T \Pi_1 E_k, W(k,r))\leq 0.1\] and \[\tv(E_k^T \Pi_2 E_k, W(k,r+2))\leq 0.1.\]
By the triangle inequality, \[\tv(E_k^T \Pi_1 E_k, E_k^T \Pi_2 E_k) \leq 0.2 + \tv(W(k,r), W(k,r+2)),\] which in turn is bounded by $0.4$ by Lemma~\ref{lem:wisharts_close_in_tvd} for $k< r/10$.  The result follows.
\end{proof}

\section{Faster sketching}
\label{appendix:faster_sketches}

In this section, we make several observations, which allow for our sketch to be applied more efficiently.

\subsection{Optimized runtime of dense sketches}
We observe that known results for fast rectangular matrix multiplication allow for the sketch to be applied in near linear time, provided that $d$ is sufficiently large relative to $\eps.$

\cite{gall2018improved} shows that multiplication of a $d \times d^{\alpha}$ matrix and a $ d^{\alpha} \times d$ matrix, may be carried out in $O(d^{2+\gamma})$ time for any $\gamma > 0,$ for $\alpha \geq 0.32.$  Since this is known to require the same number of operations as multiplying a $ d^{\alpha} \times d$ and a $d\times d$ matrix (see \cite{le2012faster} for example), our dense Gaussian sketch may be applied in time $O(d^{2+\gamma})$ as long as the sketching dimension $k$ is bounded by $O(d^{.32}).$  Since we take $k = O(1/\gamma^2)$, our sketch may be applied in near-linear time as long as $k = 1/\gamma^2 \leq O(d^{.32})$ or equivalently when $\gamma \gtrsim d^{-0.16}.$

\subsection{Faster sketching for sparse PSD matrices} We observe that a variant of our sketch may be applied quickly to sparse matrices, at least when the input matrix is PSD.

Suppose without loss of generality that $\norm{A}{F}=1.$ Our first step is to apply the $\ell_2$ heavy hitters sketch, $SAT^T$ of \cite{andoni2013eigenvalues}.  While they choose $S$ and $T$ to be Gaussian, it can be verified that their analysis carries through as long as $S$ and $T$ are $\eps$-distortion oblivious subspace embeddings on $k$ dimensional subspaces. We choose to take $S$ and $T$ to be the sparse embedding matrices of \cite{cohen2015optimal}.

Since $S$ and $T$ are in particular $O(1)$ distortion Johnson-Lindenstrauss maps, $\norm{SAT^T}{F} \leq 2\norm{A}{F}$ with good probability. Now, by setting $k = \poly(1/\eps)$ in theorem 1.2 of \cite{andoni2013eigenvalues}, we get that the singular values of $S A T^T$ approximate the top $1/\eps^2$ eigenvalues of $A$ to within $\eps$ additive error (the remaining eigenvalues of $A$ are $O(\eps)$ and so may be estimated as $0$).

Write $M = S A T^T.$  It now suffices to estimate the singular values of $M$ to $O(\eps)$ additive error.  For this we first symmetrize $M$ forming the matrix
\begin{equation}
    M_{\text{sym}} =
    \begin{pmatrix}
    0 & M\\
    M^T & 0
    \end{pmatrix}.
\end{equation}
Note that the eigenvalues of $M_{\text{sym}}$ are precisely the singular values of $M.$  To approximate the eigenvalues of $M_{\text{sym}}$ we use our dense Gaussian sketch, yielding the optimal sketching dimension of $O(1/\eps^2).$ Since $M_{\text{sym}}$ has dimensions $\poly(1/\eps)$, this last sketch may be carried out in $\poly(1/\eps)$ time.

Since $S$ and $T$ were chosen to be sparse embedding matrices, the full sketch runs in $\poly(\frac{1}{\eps}) \text{nnz}(A)$ time.  To summarize, our final sketching dimension is $O(1/\eps^2)$ on each side, and we approximate all eigenvalues to within $\eps\norm{A}{F}$ additive error.

\end{document}